\documentclass{eptcs}
\usepackage[utf8]{inputenc}
\usepackage{amsmath,amssymb,hyperref,listings,graphicx,array,tikz, dirtree, stmaryrd}

\usepackage{amsfonts,amscd,amsthm,xspace}

\usepackage{booktabs}

\title{Bridging Threat Models and Detections: Formal Verification via CADP}
\def\titlerunning{Bridging Threat Models and Detections: Formal Verification via CADP}
\author{Dumitru-Bogdan Prelipcean
\institute{
Bitdefender, Iași, Romania\\
Alexandru Ioan Cuza University, Iași, Romania\\
LACL, Université Paris-Est Créteil, France}
\email{bprelipcean@bitdefender.com}
\and 
Cătălin Dima
\institute{LACL, Université Paris-Est Créteil, France}
\email{dima@u-pec.fr}
}

\def\authorrunning{D.B. Prelipcean et al.}

\theoremstyle{plain}
\newtheorem{theorem}{Theorem}

\newtheorem{proposition}[theorem]{Proposition}

\theoremstyle{definition}
\newtheorem{definition}{Definition}
\theoremstyle{remark}

\begin{document}

\maketitle

\begin{abstract}
Threat detection systems rely on rule-based logic to identify adversarial behaviors, yet the conformance of these rules to high-level threat models is rarely verified formally. We present a formal verification framework that models both detection logic and attack trees as labeled transition systems (LTSs), enabling automated conformance checking via bisimulation and weak trace inclusion. Detection rules specified in the Generic Threat Detection Language (GTDL, a general-purpose detection language we formalize in this work) are assigned a compositional operational semantics, and threat models expressed as attack trees are interpreted as LTSs through a structural trace semantics. Both representations are translated to LNT, a modeling language supported by the CADP toolbox. This common semantic domain enables systematic and automated verification of detection coverage. We evaluate our approach on real-world malware scenarios such as LokiBot and Emotet and provide scalability analysis through parametric synthetic models. Results confirm that our methodology identifies semantic mismatches between threat models and detection rules, supports iterative refinement, and scales to realistic threat landscapes.
\end{abstract}


\section{Introduction}

Modern threat detection operates in a rapidly evolving landscape, making the correctness of detection rules critical. Analysts derive threat descriptions from natural language intelligence reports, while detection systems implement operational rules in automated engines. This disconnect often leads to incomplete coverage: for instance, a credential-stealing malware manipulating registry keys might be inadequately captured due to semantic gaps.

Cyber-threat descriptions appear in three dominant flavors.  
\emph{(i) Indicator-centric lists}—e.g., STIX/CybOX observables or simple CSV feeds—enumerate hashes, domain names, or registry paths observed in the wild~\cite{stix}.  
\emph{(ii) TTP catalogs} abstract from raw indicators to attacker behavior; the MITRE ATT\&CK matrix, CAPEC patterns, and the Lockheed Martin kill chain are canonical exemplars~\cite{mitre}.  
\emph{(iii) Structured threat models} such as attack graphs and attack trees capture causal dependencies and attacker goals, enabling what-if reasoning and quantitative risk analysis~\cite{schneier1999attack,jhawar2015stochastic}.  
Most intelligence reports published by CERTs and security vendors mix these forms: prose narrative plus ATT\&CK references, backed by Indicators of Compromise (IoCs). While expressive, such hybrid descriptions are informal and therefore unsuited for automated correctness checks.

Operational detection logic is codified in a broad ecosystem of domain-specific languages (DSLs): YARA byte-signature rules for file scanning~\cite{yara}; Snort/Suricata and Zeek policies for network intrusion detection systems (IDS)~\cite{snort}; Sigma and EQL for SIEM/XDR log analytics; and proprietary query languages such as KQL (Azure Sentinel) or SPL (Splunk).  
These languages vary in syntax but share a declarative core: predicates over event streams, optional state tracking, and actions (alert, drop, quarantine). Recent work proposes higher-level formalisms—e.g., DSLs that compile into multiple back-ends~\cite{sigma} or graph-based correlation engines—but validation still relies on executing rules against telemetry. 

\textbf{Threat descriptions} are usually validated only informally: peer review by analysts, cross-referencing with observed incidents, or replaying red-team scenarios. Formal soundness proofs for ATT\&CK-style narratives are rare; some studies apply logic-based consistency checks to attack trees~\cite{pinchinat} or use probabilistic model checking for quantitative risk~\cite{jhawar2015stochastic}.  

\textbf{Detection rules}, in turn, are tested empirically, e.g., against malware sandboxes (Cuckoo), open telemetry corpora (Mordor), or atomic red-team frameworks such as MITRE CALDERA and Atomic Red Team. Vendors run regression suites, but these provide no \emph{formal} guarantee that a rule faithfully captures the intended threat behavior—false negatives linger until real incidents expose them.

\textbf{Our contribution.} We propose a rigorous methodology that bridges this gap using formal verification. Our approach comprises three phases: (1) modeling threats with \emph{attack trees}, (2) specifying detection logic in the \emph{Generic Threat Detection Language} (GTDL), and (3) translating both into \emph{LNT}~\cite{lnt}, a formal language for modeling and verifying distributed systems. The resulting models are then analyzed automatically with the CADP toolbox~\cite{garavel2011cadp}.

Both GTDL rules and attack trees are given formal semantics as labeled transition systems (LTSs). Conformance is checked via trace inclusion or bisimulation between their LTSs, with LNT providing the common semantic domain for CADP’s model checking tools. A shared observable alphabet and semantics-preserving translation rules ensure that detection signatures can be automatically verified against the adversarial behaviors described by attack trees.

We implemented this workflow in a command-line tool that automates translation and verification. Experiments on real-world malware (LokiBot, Emotet) and on scalable synthetic models demonstrate that the approach identifies coverage gaps, supports rule refinement, and scales to realistic threat landscapes.

This work builds upon a detection specification language developed at Bitdefender\cite{bd}. A simplified version, GTDL, is presented here along with its operational semantics. GTDL structures detection logic through conditions, state transitions, and external plugin calls, supporting parallel evaluation for real-time analysis. GTDL signatures evaluate system events, behaviors, and indicators, enabling metadata inclusion, state tracking, and automated responses such as alerts or mitigations.

Unlike specialized detection languages like YARA~\cite{yara}, Snort~\cite{snort}, and Sigma~\cite{sigma}, GTDL offers a structured, generic design that extends across SIEM, IDS, and XDR architectures. It supports compositional semantics, cross-platform use, and scalability.

Our key contributions are: (1) a systematic pipeline to translate GTDL detection rules and attack trees into LNT processes with shared semantics; (2) structural operational semantics for GTDL and formalized conformance checking using bisimulation and weak trace inclusion; (3) a CLI tool for GTDL-to-LNT translation with benchmark automation; (4) a benchmark suite including real-world threats (LokiBot, Emotet) and parametric scenarios; and (5) a demonstration of applicability and scalability using CADP, showing how detection gaps can be diagnosed and repaired.

\section{Related Work}

\paragraph*{Formal verification back-ends for security.}
Automated model checking tools—including SPIN,\\
mCRL\cite{bunte2019mcrl2}, FDR4 for CSP\cite{gibson2014fdr}, and CWB\cite{cleaveland1993concurrency} have a long tradition of providing safety and liveness guarantees for protocols and middleware. While their underlying process algebras (e.g., CCS, CSP, mCRL) support basic constructs—action-prefixing (`a · P`), nondeterministic choice, and parallel composition—these primitives often lead to cumbersome encodings when modeling more complex operators, such as those found in attack trees that require flexible sequencing, choice, and OR/AND semantics. In contrast, the LNT language and the CADP toolchain~\cite{garavel2011cadp,sequential} offer richer compositional constructs, enabling more natural and concise modeling of threat patterns and formal verification tasks. Recent surveys~\cite{kulik2021survey} highlight that scalability is an ongoing challenge for formal methods in security—a challenge we address in Section~\ref{sec:evaluation} via parametric and scalable families of threat and rule models.

\paragraph*{Analysis and validation of attack trees.}
Since Schneier’s seminal introduction of attack trees~\cite{schneier1999attack}, there has been significant progress in extending their expressiveness, including probabilistic models~\cite{jhawar2015stochastic}, motivation-aware extensions~\cite{schmittner2021model}, and temporal/logical operators~\cite{pinchinat}. However, the majority of existing research focuses on analyzing the internal properties of attack trees themselves—for example, checking consistency, calculating risk, or supporting what-if analyses—rather than formally verifying whether implemented defenses (such as detection rules) actually cover the behaviors described by those trees. Our work fills this important gap by formally comparing the LTS semantics of an attack tree with that of a concrete detection rule.

\paragraph*{Detection-rule validation.}
Although detection rules are routinely tested empirically in sandboxes, 
datasets, or red-team scenarios they seldom undergo formal verification. 
Existing works largely analyze malware behavior or attack-tree consistency, 
but few verify the soundness of the detection logic itself. Formal approaches to date have largely focused on malware behavior itself (for example, state-space exploration of binaries~\cite{song2012efficient}, abstraction of low-level system actions~\cite{beaucamps}, or evolutionary models of malware families~\cite{androidevolution}) rather than on the detection logic tasked with identifying those behaviors. Christodorescu et al.~\cite{cristodorescu} initiate the extraction of formal specifications from malware samples, but do not address the challenge of formally proving soundness between specification and detection rule.

\paragraph*{Bridging threat models and detection rules.}
Some recent works attempt to bridge high-level threat modeling with operational detection by mapping MITRE ATT\&CK techniques\cite{mitre} to SIEM queries, or by correlating indicators of compromise (IoCs) with IDS signatures. However, these mappings are generally manual or rely on superficial syntactic correspondence, making them brittle and hard to maintain. To the best of our knowledge, no previous framework provides an \emph{automatic, semantics-preserving} conformance check between a structured threat description and an executable detection rule. By assigning LTS semantics to both attack trees and detection rules, and leveraging CADP for automated equivalence checking, our methodology is the first to offer this level of formal assurance.

Building on these foundations, our methodology uniquely integrates the strengths of formal modeling, attack tree analysis, and rule-based detection to address the gap between high-level threat models and executable detection logic.

\section{Proposed Methodology}
We start by directly presenting our methodology, which is structured into three main phases: threat modeling, detection modeling, and formal verification. 
By integrating these phases, we ensure a rigorous and systematic approach to validate the effectiveness of detection mechanisms against specific threats. Below, we provide a detailed description of each phase.

The first step involves describing the threat or malware using \textit{Attack Trees}, a hierarchical formalism commonly used for representing the decomposition of an attack into sub-goals. 
Once the Attack Tree is defined, it is translated in the \textit{LNT} language, a process algebra supported by the \textit{CADP} (Construction and analysis of Distributed Processes) toolbox
\cite{garavel2011cadp}.

The second step is specifying the detection using the \textit{Generic Threat Detection Language} (GTDL) signatures.

Each set of GTDL signatures are then translated into an \textit{LNT} model, by relying on the operational semantics we provide for the language and the translation rules we propose.

The final step involves using the Bisimulator tool from the CADP suite to analyze and compare the translated models of attack trees, which represent threat descriptions, against the corresponding detection models. This process determines whether the models are similar/bisimilar or have trace equivalence/inclusion, ensuring that the detection mechanisms effectively capture the described threats.

To ground this methodology, we outline the transition system semantics of Attack Trees and provide rules for translating these into LNT.
We then present a structural operational semantics for GTDL and provide rules for translating GTDL specs into LNT too. 
Finally, we recall the formal definitions of (bi)simulations, trace inclusion and equivalence relations on transition systems on which 
the third step in our methodology relies and present the details of the utilization of Bisimulator tool. 

\subsection{Attack Trees}
In this subsection we recall the notion of attack trees and their labeled transition system semantics \cite{pinchinat}. 

Recall first that a labeled transition system (LTS) is a tuple 
\[
TS = (S, Act, \rightarrow, s_0)
\]
where 
$S$ is the \emph{set of states}, 
$Act$ is the set of \emph{transition labels}, 
$\rightarrow \subseteq S \times Act \times S$ is the \emph{transition relation}, 
and $s_0 \in S$ is the distinguished \emph{initial state}. 

We write $s \xrightarrow{a} t$ whenever $(s,a,t) \in \rightarrow$. 
A \emph{path} (of length $k$) in $TS$ is a sequence of transitions
\[
\rho = s_0 \xrightarrow{a_1} s_1 \xrightarrow{a_2} \dots \xrightarrow{a_k} s_k,
\]
starting from the initial state $s_0$. 
The \emph{trace} of $\rho$ is the sequence of actions 
$tr(\rho) = a_1\ldots a_k$.
The set of traces of the transition system $TS$ is denoted 
\(\text{Traces}(TS)\). As usual, we denote the set of traces over $Act$ as $Act^*$.

Transition systems can be composed using union, concatenation and parallel composition \cite{mcbook}. 
We recall briefly here just the parallel composition operator on set of traces: 
given two traces $w_1 = a_1\ldots a_k$ and $w_2 = b_1\ldots b_l$,
their parallel composition is the set of traces denoted $w_1 \| w_2$ and defined as follows:
\[
w_1 \| w_2 = \{ u_1\cdot v_1 \ldots u_{k+l} \cdot v_{k+l} \mid u_i, v_i \in Act^*, u_1\ldots u_{k+l} = w_1, v_1 \ldots v_{k+l} = w_2 \} 
\]
The construction of a transition system for the parallel composition of two given transition systems 
is a classical construction, also called "asynchronous composition" or "shuffle" in automata theory \cite{hopcroft}.
\begin{definition}[Attack Trees]
An \emph{attack tree} is a finite, rooted tree built from a set of atomic actions $Act$ using the following constructors:
\[
A ::= \mathbf{LEAF}_a \;\mid\; \mathbf{OR}(A_1,\ldots,A_n) \;\mid\; \mathbf{AND}(A_1,\ldots,A_n) \;\mid\; \mathbf{SAND}(A_1,\ldots,A_n)
\]
where $a \in Act$ and $n \geq 2$. 
\begin{itemize}
    \item $\mathbf{LEAF}_a$ represents an atomic attack step labeled by action $a$.  
    \item $\mathbf{OR}(A_1,\ldots,A_n)$ succeeds if at least one of the subtrees $A_i$ succeeds.  
    \item $\mathbf{AND}(A_1,\ldots,A_n)$ succeeds if all subtrees $A_i$ succeed, in arbitrary order (parallel composition).  
    \item $\mathbf{SAND}(A_1,\ldots,A_n)$ succeeds if all subtrees $A_i$ succeed in the strict left-to-right order.  
\end{itemize}
Since attack trees are finite and contain no recursion, each attack tree $A$ denotes a finite set of execution traces.
\end{definition}

\begin{definition}[Trace Semantics of Attack Trees]
Let \(A\) be an attack tree constructed from atomic actions in \(Act\), using the operators \(\mathbf{LEAF}_a\), \(\mathbf{OR}\), \(\mathbf{AND}\), and \(\mathbf{SAND}\).
The set of traces \(\mathcal{T}(A)\) is defined inductively as follows:
\[
\begin{aligned}
  \mathcal{T}(\mathbf{LEAF}_a) & = \{a\} \\[1ex]
  \mathcal{T}(\mathbf{OR}(A_1, \ldots, A_n)) & = \bigcup_{i=1}^n \mathcal{T}(A_i) \\[1ex]
  \mathcal{T}(\mathbf{AND}(A_1, \ldots, A_n)) & = \mathcal{T}(A_1) \ \| \ \ldots \ \| \ \mathcal{T}(A_n) \\[1ex]
  \mathcal{T}(\mathbf{SAND}(A_1, \ldots, A_n)) & = \mathcal{T}(A_1) \cdot \ldots \cdot \mathcal{T}(A_n)
\end{aligned}
\]
Here, \(\|\) denotes the shuffle (interleaving) of traces, and \(\cdot\) denotes concatenation of traces.
\end{definition}

\noindent
Intuitively, $\mathcal{T}(A)$ is the set of all possible observable action sequences (traces) corresponding to successful executions of the attack tree $A$.

\subsection*{Translation of Attack Trees into LNT}

We give here the rules for translating an attack tree into a set of LNT processes such that, whenever an atomic action is executed, a corresponding non-silent event is emitted.

\paragraph*{Traces of LNT Processes.}
The set of traces $Traces(P)$ of an LNT process $P$ is defined as the set of sequences of observable actions produced by $P$ under the standard operational semantics of LNT, that is, the sequences of labels corresponding to executions from the initial process to termination~\cite{lnt,processalgebrabook}.
For the relevant constructs:
\begin{itemize}
  \item $Traces(\texttt{process A [a:any] is a end process;}) = \{a\}$
  \item $Traces(\texttt{select } P_1 [] \cdots [] P_n \texttt{ end select;}) = \bigcup_{i=1}^n Traces(P_i)$
  \item $Traces(\texttt{par } P_1 \texttt{ || } \cdots \texttt{ || } P_n \texttt{ end par;}) = Traces(P_1) \| \cdots \| Traces(P_n)$
  \item $Traces(\texttt{P}_1; \cdots; \texttt{P}_n) = Traces(P_1) \cdot \cdots \cdot Traces(P_n)$
\end{itemize}
Here, $\|$ denotes the shuffle (interleaving) of traces, and $\cdot$ denotes concatenation.

\begin{definition}[Attack Tree to LNT Translation]
We define a translation function $tr(\cdot)$ from attack trees to LNT processes as follows:
\[
\begin{aligned}
  tr(\mathbf{LEAF}_a) \quad &= \texttt{process LEAF\_a [a: any] is a end process;} \\[2ex]
  tr(\mathbf{OR}(A_1, \ldots, A_n)) \quad &= \texttt{process OR\_A is} \\
    &\quad \texttt{select $P_{A_1}$ [] $\cdots$ [] $P_{A_n}$ end select; end process} \\[2ex]
  tr(\mathbf{AND}(A_1, \ldots, A_n)) \quad &= \texttt{process AND\_A is} \\
    &\quad \texttt{par $P_{A_1}$ || $\cdots$ || $P_{A_n}$ end par; end process} \\[2ex]
  tr(\mathbf{SAND}(A_1, \ldots, A_n)) \quad &= \texttt{process SAND\_A is} \\
    &\quad \texttt{$P_{A_1}$ ; $\cdots$ ; $P_{A_n}$ ; end process;}
\end{aligned}
\]
\end{definition}

Here, $P_{A_i}$ denotes the name of the process corresponding to the subtree $A_i$, obtained by recursively applying $tr(\cdot)$. 
Only leaf processes carry a formal parameter \texttt{[a:any]} to represent the observable action emitted; composite operators (\textbf{OR}, \textbf{AND}, \textbf{SAND}) serve only to compose their children and therefore do not themselves introduce parameters.

\begin{proposition}[Correctness of the Translation]
For every attack tree $A$, the traces of its LNT translation coincide with its trace semantics:
\[
Traces(tr(A)) = \mathcal{T}(A).
\]
\end{proposition}

\begin{proof}
The proof is immediate by structural induction on $A$. 
Each translation clause mirrors the corresponding semantic clause in the definition of $\mathcal{T}(\cdot)$:
\begin{itemize}
  \item leaves emit a single observable action, 
  \item $\mathbf{OR}$ translates to nondeterministic choice, 
  \item $\mathbf{AND}$ translates to parallel composition (shuffle), and 
  \item $\mathbf{SAND}$ translates to sequential composition (concatenation). 
\end{itemize}
Hence the equality follows directly.
\end{proof}

\subsection{Generic Threat Detection Language (GTDL)}
\label{GTDL}
The introduction and formalization of Generic Threat Detection Language (GTDL) are key contributions of this paper. GTDL is a domain-specific language for defining structured detection rules essential for real-time attack monitoring and identification.

Developed through research at Bitdefender, GTDL addresses challenges in complex threat environments. It introduces detection rules, or signatures, which define conditions for identifying attack components. Its modular architecture enhances effectiveness by enabling interaction and shared states between signatures.

GTDL signatures specify conditions for detecting cyber threats based on system events, process behaviors, and contextual indicators. An evaluation engine interprets these rules, periodically processing signals to determine if a threat condition is met. Upon detection, the engine can set global flags, generate alerts, or trigger automated mitigation actions. In our model, we focus on the initial response by setting predefined flags corresponding to the targeted attack.
 
Figure \ref{fig:GTDL2} gives the workflow of this engine and its interaction with signatures, illustrating how signals are processed, rules are evaluated, and threat responses are executed.

\begin{figure}[ht]
	\centering
	\includegraphics[scale=0.40]{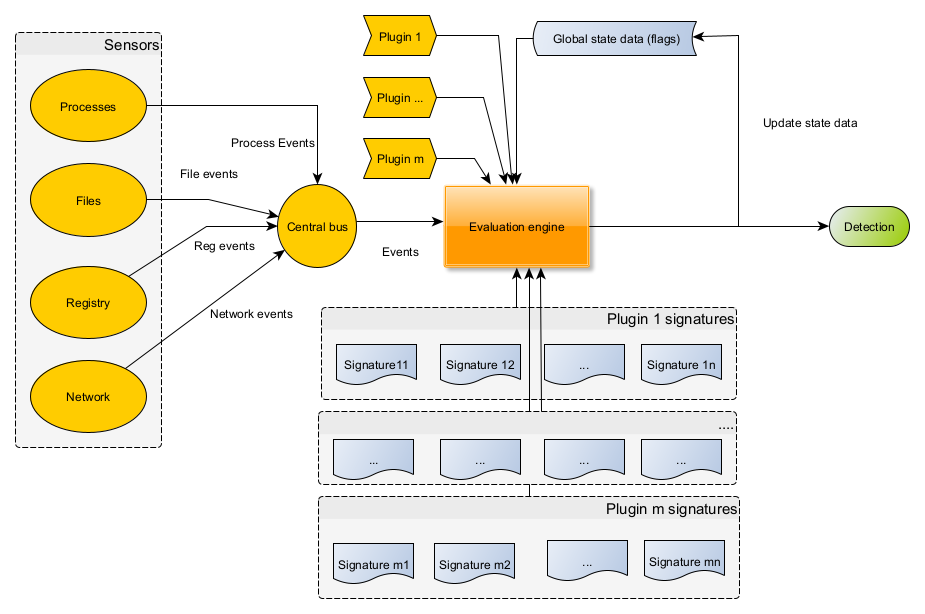}
	\caption{Workflow of the GTDL detection engine.}
	\label{fig:GTDL2}
	 \vspace*{-50pt}
 \end{figure}
The GTDL detection framework consists of several core components that work together to identify and respond to threats. The \textbf{sensors layer} continuously monitors system activities, including process execution, file modifications, registry changes, and network traffic. Data collected by the sensors is forwarded to the \textbf{central bus}, which aggregates and directs it to the \textbf{evaluation engine} for real-time analysis. The engine processes incoming data using \textbf{global state information} and integrates with \textbf{plugins and signatures}. Plugins define detection logic, while signatures contain metadata and rules that specify conditions for identifying malicious behavior. Together, these components enable an efficient and scalable threat detection mechanism.

The detection process follows a structured workflow to identify and respond to potential threats. It begins with the \textbf{sensors} layer, which captures system events such as process executions, file modifications, and network activities, relaying this data to the appropriate \textbf{plugins}. These plugins process the raw event data, extract relevant security indicators, and submit the processed information to the \textbf{evaluation engine}. The evaluation engine then applies \textbf{detection rules} defined in GTDL signatures, leveraging both real-time data and historical state tracking. When the specified conditions in a \textbf{detection signature} are met, the system flags the corresponding threat, triggering appropriate response actions such as alerting security teams or initiating automated mitigation measures. This structured approach ensures efficient and accurate threat detection.
\paragraph{GTDL Syntax.}
We first describe the core syntax of a single GTDL rule. 
Each detection rule is expressed as a finite program built from assignments, conditionals, and a detection action:
\begin{align*}
sigdet &::= \textbf{[DETECTION]}~\texttt{Name = } 'D'~\ldots~\textbf{[RULE]}~block \\
block  &::= assign~;~block \mid stmt \\
assign &::= v := inPluginCall(f, arg) \mid v := e \\
stmt   &::= \textbf{IF}~cond~\textbf{THEN}~block~\textbf{ELSE}~block~\textbf{END IF} \mid action \\
cond   &::= v \mid e \mid cond~\textbf{AND}~cond \mid cond~\textbf{OR}~cond \mid \textbf{NOT}~cond \mid v == e \mid v \neq e \\
e      &::= \textbf{true} \mid \textbf{false} \mid var \\
action &::= \texttt{GlobalFlag.Set("D")} \\
\end{align*}

Here $v$ ranges over Boolean variables, and $inPluginCall(f,arg)$ invokes a plugin API that returns a Boolean. 
The detection action \texttt{GlobalFlag.Set("D")} emits a detection for~$D$.

The grammar above describes the syntax of an individual rule. 
At runtime, the detection engine manages multiple rules simultaneously: 
\begin{itemize}
    \item Each rule is executed as a separate process, and the engine runs these processes \emph{in parallel}. 
    \item Input events (e.g., process creation, file access) are streamed to all rules. 
    \item Rules are re-executed (looped) on each new event, but this iteration belongs to the execution environment, not the syntax of GTDL itself. 
\end{itemize}

Thus, although the syntax has no explicit parallel composition or looping constructs, these are provided implicitly by the detection engine running many rules continuously over a stream of system events.

\paragraph{Structural operational semantics.}

We formalize the semantics of a \emph{single} GTDL rule.
A \emph{configuration} is a pair $(C,\sigma)$ where $C$ is a command (block or statement) and
$\sigma : \textsf{Var} \to \mathbb{B}$ is a store mapping Boolean variables to the Boolean domain
$\mathbb{B}=\{\mathsf{tt},\mathsf{ff}\}$.
We write unlabeled steps as $\rightarrow$ and use $\xrightarrow{d}$ only for the observable
\emph{detection} action associated to $D$.
The special command \texttt{halt} has no outgoing transitions.

\begin{table}[ht]
  \centering
  \renewcommand{\arraystretch}{1.25}
  \begin{tabular}{|p{4.1cm}|p{10cm}|}
    \hline
    \textbf{Rule} & \textbf{Transition} \\ \hline
    Assignment (plugin) &
    $(\,v := \textit{inPluginCall}(f,arg);\;C,\;\sigma\,)\ \rightarrow\
     (\,C,\;\sigma[v \mapsto \mathrm{pluginValue}(f,arg)]\,)$ \\ \hline
    Assignment (expr) &
    $(\,v := e;\;C,\;\sigma\,)\ \rightarrow\
     (\,C,\;\sigma[v \mapsto \llbracket e \rrbracket_\sigma]\,)$ \\ \hline
    If-True &
    $(\,\textbf{IF}\ c\ \textbf{THEN}\ B_1\ \textbf{ELSE}\ B_2\ \textbf{END IF},\;\sigma\,)\ \rightarrow\
     (\,B_1,\;\sigma\,)$ \quad if $\llbracket c \rrbracket_\sigma=\mathsf{tt}$ \\ \hline
    If-False &
    $(\,\textbf{IF}\ c\ \textbf{THEN}\ B_1\ \textbf{ELSE}\ B_2\ \textbf{END IF},\;\sigma\,)\ \rightarrow\
     (\,B_2,\;\sigma\,)$ \quad if $\llbracket c \rrbracket_\sigma=\mathsf{ff}$ \\ \hline
    Detection Action &
    $(\,\texttt{GlobalFlag.Set("D")},\;\sigma\,)\ \xrightarrow{d}\ (\,\texttt{halt},\;\sigma\,)$ \\ \hline
  \end{tabular}
  \caption{Small-step rules for a single GTDL rule}
  \label{tab:sos_rules_gtdl_core}
\end{table}

The evaluation function $\llbracket\cdot\rrbracket_\sigma$ maps expressions/conditions to $\mathbb{B}$:

\begin{table}[ht]
  \centering
  \renewcommand{\arraystretch}{1.25}
  \begin{tabular}{|p{6.4cm}|p{7.7cm}|}
    \hline
    \textbf{Syntactic form} & \textbf{Semantic value} \\ \hline
    $\llbracket v \rrbracket_\sigma$ & $\sigma(v)$ \\
    $\llbracket \texttt{true} \rrbracket_\sigma$ & $\mathsf{tt}$ \\
    $\llbracket \texttt{false} \rrbracket_\sigma$ & $\mathsf{ff}$ \\
    $\llbracket \textit{inPluginCall}(f,arg) \rrbracket_\sigma$ & $\mathrm{pluginValue}(f,arg) \in \mathbb{B}$ \\ \hline
    $\llbracket c_1\ \textbf{AND}\ c_2 \rrbracket_\sigma$ & $\llbracket c_1 \rrbracket_\sigma \wedge \llbracket c_2 \rrbracket_\sigma$ \\
    $\llbracket c_1\ \textbf{OR}\ c_2 \rrbracket_\sigma$ & $\llbracket c_1 \rrbracket_\sigma \vee \llbracket c_2 \rrbracket_\sigma$ \\
    $\llbracket \textbf{NOT}\ c \rrbracket_\sigma$ & $\neg \llbracket c \rrbracket_\sigma$ \\ \hline
    $\llbracket v == e \rrbracket_\sigma$ & $(\,\sigma(v) = \llbracket e \rrbracket_\sigma\,)$ \\
    $\llbracket v \neq e \rrbracket_\sigma$ & $(\,\sigma(v) \neq \llbracket e \rrbracket_\sigma\,)$ \\ \hline
  \end{tabular}
  \caption{Valuation of expressions and conditions}
  \label{tab:valuation_gtdl}
\end{table}

Here $\mathrm{pluginValue}(f,arg)$ denotes the Boolean result produced by the plugin on the current input event,
and $d$ is the observable action label corresponding to detection $D$.
A \emph{trace} is the (finite) sequence of observable labels emitted during an execution; with the rules above, only $d$ can appear in traces.
The rules in Table~\ref{tab:sos_rules_gtdl_core} define the stepwise execution of a single GTDL rule. 
In practice, the detection engine manages a set of such rules as follows:
\begin{itemize}
  \item Each rule is executed as an independent process, and the engine runs all rules \emph{in parallel}. 
  \item Input events are broadcast to all rules; each rule updates its variables and evaluates conditions accordingly. 
  \item The engine repeatedly re-executes rules on successive input events, which provides the effect of looping over the event stream. 
\end{itemize}

Thus, while the core language syntax has no explicit constructs for parallel composition or looping, these arise naturally from the execution model: multiple rules are run in parallel and continuously re-applied to the incoming event stream.

\paragraph*{The model of GTDL in LNT}
Variables assigned via plugin calls become \emph{process parameters} in LNT, e.g., if the GTDL rule refers to process name and process path, then the LNT process receives these as input arguments. The detection logic (assignments, conditions, actions) is translated homomorphically into the body of the LNT process. The detection output $GlobalFlag.Set("D")$ is translated to an output on channel $dSet$.

\[
\begin{aligned}
\mathcal{T}(\textbf{[DETECTION]}~\texttt{Name = } 'D'~\ldots~\textbf{[RULE]}~block) =\ & \texttt{process } D [dSet:\texttt{FLAG\_CHANNEL}]~(\mathcal{P})~\texttt{is} \\
& \qquad \mathcal{T}(block) \\
& \texttt{end process}
\end{aligned}
\]
where $\mathcal{P}$ is the (ordered) tuple of input variables corresponding to all $inPluginCall$ instances in $block$.

For the translation $\mathcal{T}(block)$:
\begin{table}[ht]
    \centering
    \renewcommand{\arraystretch}{1.5}
    \begin{tabular}{|p{5.6cm}|p{8.5cm}|}
        \hline
        \textbf{GTDL Construct} & \textbf{LNT Translation $\mathcal{T}(\cdot)$} \\ \hline
        $v := inPluginCall(f, arg);~rest$ &
        Process parameter for $v$ (do not generate assignment; $v$ is an input parameter). Translate $rest$. \\ \hline
        $v := e;~rest$ &
        \texttt{var} $v := \mathcal{T}(e);$ $\mathcal{T}(rest)$ \\ \hline
        $\textbf{IF}\ cond\ \textbf{THEN}\ B_1\ \textbf{ELSE}\ B_2\ \textbf{END IF}$ &
        \texttt{if} $\mathcal{T}(cond)$ \texttt{then} $\mathcal{T}(B_1)$ \texttt{else} $\mathcal{T}(B_2)$ \texttt{end if} \\ \hline
        $\texttt{GlobalFlag.Set("D")}$ &
        \texttt{dSet(TRUE)} \\ \hline
        $cond_1~\textbf{AND}~cond_2$ &
        $\mathcal{T}(cond_1)~\texttt{and}~\mathcal{T}(cond_2)$ \\ \hline
        $cond_1~\textbf{OR}~cond_2$ &
        $\mathcal{T}(cond_1)~\texttt{or}~\mathcal{T}(cond_2)$ \\ \hline
        $\textbf{NOT}~cond$ &
        \texttt{not} $\mathcal{T}(cond)$ \\ \hline
        $v == e$ &
        $v~\texttt{==}~\mathcal{T}(e)$ \\ \hline
        $v \neq e$ &
        $v~\texttt{<>}~\mathcal{T}(e)$ \\ \hline
        $\textbf{[DETECTION]}$ $D$ $[...]$ $[RULE]~block$ &
        \begin{minipage}[t]{8.3cm}
        \texttt{process} $D$ \texttt{[dSet:FLAG\_CHANNEL]} $(\mathcal{P})$ \texttt{is}\\
        \qquad $\mathcal{T}(block)$\\
        \texttt{end process}\\
        where $\mathcal{P}$ is the tuple of variables assigned via plugin calls.
        \end{minipage} \\ \hline
    \end{tabular}
    \caption{Translation from GTDL to LNT}
    \label{tab:translation_gtdl_lnt}
\end{table}

\textbf{Parallel signatures:} If the GTDL detection engine has multiple rules, their LNT translations are composed in parallel.

Let $\llbracket P \rrbracket_\text{GTDL}$ be the set of traces of a GTDL rule $P$ (i.e., sequences of detection actions for all possible input streams), and $Traces(\mathcal{T}(P))$ be the traces of its LNT translation.

\begin{theorem}
For any GTDL detection rule $P$, $Traces(\mathcal{T}(P)) = \llbracket P \rrbracket_\text{GTDL}$.
\end{theorem}

\begin{proof}
By induction on the structure of $P$:
\begin{itemize}
    \item \textbf{Base case:} For an action $GlobalFlag.Set("D")$, both GTDL and LNT emit $d$ or $dSet(TRUE)$ respectively.
    \item \textbf{Assignment:} Assignments from plugin calls are mapped to process parameters; all other assignments are preserved in LNT as local variables, so the variable environments match.
    \item \textbf{Conditionals:} IF/THEN/ELSE structures are translated to LNT `if/then/else` blocks, preserving branch logic and thus trace outcomes.
    \item \textbf{Boolean connectives:} Boolean logic (AND, OR, NOT, ==, $\neq$) is mapped homomorphically to LNT.
    \item \textbf{Parallel composition:} Multiple GTDL rules are composed in parallel in LNT, so traces interleave as in GTDL.
\end{itemize}
Thus, for any input, the sequence of outputs (detection actions) is the same in GTDL and its LNT translation.
\end{proof}

\subsection{Formal Verification}

The last step in our methodology is to utilize the Bisimulator tool to check that some detection signature detects a malicious behavior corresponding with an attack tree, 
by feeding the transition systems generated for both to the Bisimulator tool \cite{bisimulator}.
Recall that Bisimulator can check strong or weak simulation and bisimulation, as well as trace inclusion/equivalence and weak trace inclusion/equivalence. 
Hence, our methodology interprets the "correspondence" between a signature specification and an attack tree as any of these relations between their transition systems. 
For the sake of completeness, we recall briefly here the notions of (strong and weak) simulation and bisimulation and of weak traces associated to a transition system.

\begin{definition}[Simulation/bisimulation]
Given two transition systems $\text{LTS}_1 = (S_1, \Sigma, \rightarrow_1, s_1^0)$ and $\text{LTS}_2 = (S_2, \Sigma, \rightarrow_2, s_2^0)$,
a \emph{simulation} is a binary relation $\mathcal{R} \subseteq S_1 \times S_2$
satisfying the following two properties:
\begin{enumerate}
    \item For each pair of states $(s_1,s_2) \in \mathcal{R}$, if 
    $s_1 \xrightarrow{a}_1 s_1'$ for some $s_1'\in S_1$ then there exists $s_2' \in S_2$ with $s_2 \xrightarrow{a}_2 s_2'$ and $(s_1', s_2') \in \mathcal{R}$.
    \item $ (s_1^0, s_2^0) \in \mathcal{R}$, 
\end{enumerate}
A \emph{bisimulation} is a simulation $\mathcal{R} \subseteq S_1 \times S_2$ such that $\mathcal{R}^{-1}$ is a simulation too.
\end{definition}

We say that $\text{LTS}_1$ \textbf{simulates} $\text{LTS}_2$, denoted $\text{LTS}_1 \preceq \text{LTS}_2$, 
whenever there exists a simulation $\mathcal{R}$ between the two systems. 
When $\mathcal{R}$ is actually a bisimulation, we say that $\text{LTS}_1$ is \textbf{bisimilar with} $\text{LTS}_2$ and denote  $\text{LTS}_1 \sim \text{LTS}_2$.

\begin{definition}[Weak Simulation \& Bisimulation]
Assume that the action set $\Sigma$ of both systems contains a special action $\tau$ (representing internal, non-observable actions).
A \emph{weak simulation} 
is a relation $\mathcal{R} \subseteq S_1 \times S_2$ satisfying the following three properties:
\begin{enumerate}
    \item For each pair of states $(s_1,s_2) \in \mathcal{R}$, if 
    $s_1 \xrightarrow{a}_1 s_1'$ for some $s_1'\in S_1$ then there exists $s_2', s_2'', s_2''' \in S_2$ with $s_2 \xrightarrow{\tau^*} s_2' \xrightarrow{a}_2 s_2'' \xrightarrow{\tau^*} s_2'''$ and $(s_1', s_2''') \in \mathcal{R}$.
    \item For each pair of states $(s_1,s_2) \in \mathcal{R}$, if 
    $s_1 \xrightarrow{\tau}_1 s_1'$ for some $s_1'\in S_1$ then there exists $s_2' \in S_2$ with $s_2 \xrightarrow{\tau^*} s_2'$ and $(s_1', s_2') \in \mathcal{R}$.
    \item $ (s_1^0, s_2^0) \in \mathcal{R}$, 
\end{enumerate}

A \emph{weak bisimulation} is a weak simulation $\mathcal{R} \subseteq S_1 \times S_2$ such that $\mathcal{R}^{-1}$ is a weak simulation too.

\end{definition}

Similarly to the case of (non-weak) simulations, we say that $\text{LTS}_1$ \textbf{weakly simulates} $\text{LTS}_2$, denoted $\text{LTS}_1 \precapprox \text{LTS}_2$, 
whenever there exists a simulation relation $\mathcal{R}$ between the two systems. 
When $\mathcal{R}$ is actually a bisimulation, we say that $\text{LTS}_1$ is \textbf{weakly bisimilar with} $\text{LTS}_2$ and denote  $\text{LTS}_1 \approx \text{LTS}_2$.

The Bisimulator tool can check trace inclusion, trace equivalence, and related properties over both standard and \emph{weak traces}. 
In our encoding, silent actions (\(\tau\)) represent internal GTDL bookkeeping (assignments, variable updates, condition checks) 
and are abstracted away, so that only observable detections remain under weak equivalence.
In our context, \emph{trace inclusion} guarantees that all adversarial behaviors 
specified in the attack tree are detected (i.e., no false negatives). 
By contrast, \emph{bisimulation} additionally rules out over-approximation, 
ensuring that the detection logic does not generate traces beyond those 
described by the threat model (i.e., no false positives).

Consider the morphism $\varphi : \Sigma^* \longrightarrow \Sigma^* $ defined by $\varphi(\tau) = \epsilon$ and $\varphi(a) = a$ for $a\neq \tau$.
We say that two traces $w_1,w_2 \in \Sigma^*$ are \textbf{weakly equivalent}, denoted $w_1 \simeq w_2$, if $\varphi(w_1) = \varphi(w_2)$.
A \textbf{weak trace} is then an equivalence class for the (equivalence) relation $\simeq$. 
Two transition systems  $\text{LTS}_1$  and $\text{LTS}_2$ are then weak trace equivalent if their sets of weak traces are the same.

\paragraph{Ensuring a common set of actions for both attack trees and detection models}
(Bi)simulations and trace equivalences/inclusions between two labeled transition systems can only be checked when both systems share the same set of actions. 
This important requirement is implemented in our approach by the existence of a unique channel of observable actions. Recall that Bisimulator considers that internal actions in LNT processes as the equivalent of the silent action $\tau$, 
and all other actions as observable. A bijective correspondence between actions executed by the LNT translation of an attack tree 
and system events detected by GTDL signatures
is then ensured as follows:

\begin{itemize}
\item The LNT translation of both the attack tree and the GTDL detection utilize a communication channel which records all events occurring during their execution.  
    \item The \textbf{attack tree model} emits an event when an action (labeling a leaf node) is executed performed. \begin{lstlisting}
process AttackLeaf_a [flag_a:FLAG_CHANNEL] is 
    flag_a(TRUE);
end process
\end{lstlisting} \texttt{flag(TRUE)} indicates that the attack step has been performed.
    \item The \textbf{detection model} emits an event when an action is \textbf{detected}.\begin{lstlisting}
process Signature [flag_a:FLAG_CHANNEL] (in var x:any) is 
    if detect_action_a(x) then flag_a(TRUE);
end process
\end{lstlisting} When \texttt{detect\_action\_a(x)} holds (placeholder for any condition), GTDL registers the action by setting the flag.
\end{itemize}

This ensures that both models share the same alphabet (observable action), enabling the use of Bisimulator for verifying the above mentioned relations.

 In operational settings, detection events are often partial, noisy, or ambiguous, which may lead to false positives or negatives. Addressing this challenge requires extending our methodology to handle probabilistic or fuzzy conformance (e.g., via probabilistic model checking or trace abstraction), and relaxing strict equivalence to tolerance-based metrics. We consider this an important avenue for future research.

 It needs to be mentioned that  our methodology relies on the availability and quality of attack trees as threat models. Constructing attack trees for complex or novel systems is a challenging and often manual process, requiring extensive domain expertise and system-specific knowledge~\cite{schneier1999attack,pietre}. Moreover, standardized or public datasets of attack trees are rare, and existing trees are typically tailored to specific infrastructures or incident histories. As such, our formal guarantees are contingent on the threat model accurately reflecting relevant adversarial behavior; the verification process cannot compensate for incomplete or imprecise modeling. Future work should explore techniques for semi-automated attack tree synthesis, model validation, and integration of alternative threat modeling frameworks (e.g., kill chains, MITRE ATT\&CK graphs) to alleviate this bottleneck.

\section{Evaluation}
\label{sec:evaluation}
We evaluate our methodology on eight real–world intrusion scenarios. For every case study we \emph{start from the production detection content that ships with a commercial EDR/SIEM platform actually deployed}. We then apply our formal workflow to highlight blind spots and iteratively patch the rules. Each subsection below narrates the attack progress, the vendor-supplied detection baseline, the improvements suggested by verification, and the final equivalence result; example of artefacts are deferred to the appendix with technical listings (\autoref{sec:appendix_listings}).
\subsection*{Verification Relations and Their Interpretation}

Once both the attack tree and the detection rules are translated into LNT, we
compare their labelled transition systems using equivalence or inclusion
relations provided by CADP. These relations have distinct interpretations:

\begin{itemize}
  \item \textbf{Strong bisimulation} requires an exact stepwise correspondence
  between the attack model and the detection logic. This is often too strict in
  practice, since it forces every internal step of the detection engine to
  match the threat model.

  \item \textbf{Weak (observational) bisimulation} abstracts away from silent
  actions ($\tau$), which in our encoding correspond to assignments and
  condition evaluations. \\
  Only detection actions
  (\texttt{GlobalFlag.Set("D")}) remain observable.

  \item \textbf{Trace equivalence} requires equality of observable traces:
  the detection logic accepts exactly the same traces as the attack tree, no
  more and no less.

  \item \textbf{Trace inclusion} requires that every adversarial behavior
  described by the attack tree is covered by some detection trace. The
  detection may still over-approximate by producing additional traces, which is
  acceptable from a security standpoint.
\end{itemize}

\paragraph*{Interpretation.}
If trace inclusion holds, then all adversarial behaviors specified in the
attack tree are guaranteed to be detected. If equivalence (or bisimulation)
holds, the detection matches the model exactly, with no false negatives and no
over-approximation. If even inclusion fails, then some threats described by the
attack tree are \emph{not} covered by the detection logic. Silent actions are
chosen precisely to model internal bookkeeping in GTDL (assignments, variable
updates), ensuring that verification focuses only on externally observable
detections.

\paragraph{Tool chain.}
The tool\cite{mindthegap} provides an integrated set of scripts to automate the verification workflow. The script \texttt{tree2lnt.py} translates YAML-specified attack trees (\texttt{AND}, \texttt{OR}, \texttt{SAND}) into LNT processes, while \texttt{gtdl2lnt.py} converts corresponding GTDL detection rules into LNT monitors. The \texttt{verify.sh} script compiles both models with CADP, performs LTS minimization, and invokes the \texttt{bisimulator} tool to decide conformance using strong, observational, or weak-trace equivalence. To support scalability analysis, \texttt{measure\_times.py} iterates over all \texttt{benchmark\_*} folders, runs the verification pipeline, and records timing data as shown in Table~\ref{tab:execution_times}. Additionally, \texttt{parametric/generator.py} can be used to generate synthetic attack trees for large-scale experiments. The toolchain requires only Python 3.9 and CADP 2024-x/2025-x for operation.

\subsection{Case Study: Detection and Verification Across Attack Campaigns}
The following table provides a comparative overview of representative malware cases examined in this work, summarizing both their technical behaviors and the corresponding detection strategies. For each case, we highlight the attack narrative, key steps or tactics leveraged by adversaries, and the detection rules or correlation logic employed. Where available, we also note how formal verification methods (based on CADP's Bisimulator counterexamples) shaped detection coverage and rule refinement. This structured summary facilitates rapid cross-case comparison and illustrates the application of formal methods to practical detection engineering.

\begin{table}[htbp]
\centering
\renewcommand{\arraystretch}{1.2}
\small
\begin{tabular}{|p{2.0cm}|p{3.2cm}|p{4.5cm}|p{4.8cm}|}
\hline
\textbf{Case} & \textbf{Attack/Narrative} & \textbf{Key Steps/Behavior} & \textbf{Detection Logic \& Verification} \\
\hline
\textbf{LokiBot} \cite{lokiBot} 
& Infostealer via phishing attachments
& Renamed process in \%TEMP\%, auxiliary payloads, credential exfiltration
& Correlates four partial signatures: process name, malicious extension, run-key, C\&C; An AND-structured attack tree. \\

\hline
\textbf{Emotet} \cite{dfir_emotet}
& Loader for multi-stage intrusions, ending in Quantum ransomware
& Steps: some strictly ordered (lateral movement), some parallel (discovery)
& Mix of AND/SAND operators; iterative verification expanded ruleset to over-approximate attack. \\

\hline
\textbf{BlackCat (ALPHV)} \cite{mandiant_blackcat22,sophos_blackcat22,ncc_blackcat23,dfir_blackcat24}
& Double-extortion ransomware; rapid cross-platform attacks in Rust
& VPN credential reuse, Kerberos ticket forging, payload drop, shadow copy deletion, encryption (delayed)
& Five uncorrelated EDR flags (VPN anomaly, private hash, payload drop, Rust process, shadow delete); CADP highlighted coverage gaps; formal verification and improved rule set. \\

\hline
\textbf{Cyber Av3ngers} \cite{cisa_ics24,dragos_avengers23}
& ICS/OT campaign targeting PLCs
& Scan PLCs, unauthorized access, alarm disable, destructive shutdown
& Independent detection signatures mapped to alarm messages and PLC traffic; behavior-mapped signatures. \\

\hline
\textbf{Gamaredon} \cite{unit42_gamaredon22,eset_gamaredon24}
& Rapid-fire phishing ops with macros and staged payloads
& Delivery doc, macro execution, staged PowerShell downloaders, C2 beacon
& Partial signature tracking per stage; formal verification confirmed equivalence, improved macro heuristic. \\

\hline
\textbf{IcedID} \cite{proofpoint_icedid24,ibm_icedid23}
& Loader via macro-enabled \texttt{.xlsm} attachments
& DLL injection and C2 beacon, order can vary
& Attack tree: AND; detection uses disjunction (process/file) + conjunction (C2 domain); fits unordered steps. \\

\hline
\textbf{QakBot (QBot)} \cite{cisa_qakbot24,talos_qakbot24}
& Resurgent loader/macro phishing
& Macro document delivery, execution, DLL injection, credential dump, beaconing
& Five-flag correlation for alert; verified by equivalence checks (observational/weak trace). \\
\hline
\textbf{TraderTraitor} \cite{cisa_trader22,rf_trader23}
& Supply-chain attacks on crypto traders
& Installer drop-in, persistence via scheduled tasks, updater download, remote shell
& AND root: all independent leaves; reuses EDR IOC matches, no extra instrumentation. \\
\hline
\end{tabular}
\caption{Summary of detection and verification approaches for selected malware cases.}
\label{tab:malware_cases}
\end{table}

\subsubsection{Parametric models}
In this section we address the scalability of our methodology for verifying conformance of attack tree and detection with respect tothe number of leaves/signatures, the type of operators (AND, OR, SAND), and the layers (nodes) in attack tree.
Our test methodology for measuring the scale of generated BCG \footnote{BCG (Binary Coded Graphs) is a compact storage format used in CADP for efficiently representing Labeled Transition Systems (LTSs). It enables scalable storage, manipulation, and analysis of large state spaces. BCG integrates with CADP tools like \texttt{bcg\_min} for LTS minimization and \texttt{bisimulator} for bisimulation checking. By supporting strong, weak, and branching bisimulations, it ensures efficient system verification. Its optimized algorithms facilitate fast computation, making it essential for model checking and equivalence verification in formal methods.}for each model is the following. 
For attack trees we use as parameter the number of leaves considering a different action for each leaf. 
Then we generated the event definitions, the process for each leaf and the combining node (root) where the operator construction was used.
In terms of detection we used the same approach but the differences can be found in the interpretation of the model.
\begin{figure}[ht]
	\centering
	\includegraphics[scale=0.5]{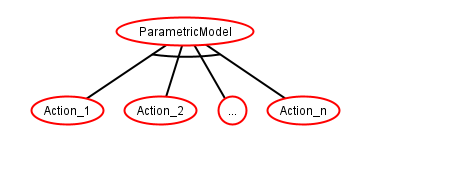}
    \vspace*{-35pt}    
	\caption{Parametric attack tree}
	\label{fig:paramAT}
     \vspace*{-15pt}  
\end{figure}
A model for detection for an attack tree with 3 leaves is as follows:
\begin{flushleft}
\small
\begin{verbatim}
process Engine [detectionFlag1,
    detectionFlag2,detectionFlag3:FLAG_CHANNEL, 
    finalDetection:any] is
par
SignatureDetectionFlag1[detectionFlag1] ||
SignatureDetectionFlag2[detectionFlag2] ||
SignatureDetectionFlag3[detectionFlag3] 
end par	
end process
\end{verbatim}
\end{flushleft}
The case studies are described further.
We analyze the execution time of equivalence verification under different composition schemes: AND-only, SAND-only, OR-only, AND-OR, and AND-SAND.  

For the \textbf{AND-only} case, models are constructed using only parallel composition. The simplest case includes leaves/signatures that emit observable actions without intermediary computations. The model size grows exponentially, making execution time crucial. Our results show that weak trace is significantly faster due to the deterministic transition system, avoiding exponential operations. Memory usage is lower for weak trace, with a smaller variable size (16 bytes vs. 24 bytes), fewer visited variables, and a more efficient DFS stack size of \( n+1 \) (compared to \( n+2 \) for observational). This suggests weak trace is preferable for verifying detection models against attack specifications in large systems.  
For the \textbf{non-deterministic AND-only} case (where events appear in two leaves), execution time increases but weak trace remains faster.  

For the \textbf{SAND-only} case, models use sequential action modeling. The model size and execution time increase linearly, with no significant differences between the observational and weak trace methods.  
The \textbf{OR-only} case is modeled via non-deterministic choice, leading to a linear increase in model size. Execution times are similar for both methods, making them interchangeable.  

The \textbf{AND-OR} composition includes both parallel (AND) and non-deterministic (OR) components. We kept two intermediary OR nodes with three leaves and computed weaktrace inclusion similarly.  
For the \textbf{AND-SAND} case, models combine parallelism (AND) with sequential constraints (SAND). Instead of observational bisimulation, we used pre-order simulation, which verifies that detection models cover the attack tree. We maintained two SAND nodes with three leaves while computing weaktrace inclusion.  

Table \ref{tab:execution_times} summarizes execution times for all cases. Weaktrace remains consistently faster, particularly in larger models, making it a preferred option for verifying detection compliance.  

\begin{table}[ht]
\centering
\setlength{\tabcolsep}{4pt} 
\caption{Equivalence execution time (seconds) for different cases}
\label{tab:execution_times}
\begin{tabular}{|c||c|c||c|c||c|c||c|c||c|c|}
 \hline
 \# lvs. & \multicolumn{2}{c||}{AND-only} & \multicolumn{2}{c||}{AND-non-det.} & \multicolumn{2}{c||}{SAND-only} & \multicolumn{2}{c||}{AND-SAND} & \multicolumn{2}{c|}{AND-OR} \\
 \cline{2-11}
  & Obs. & Wktrc. & Obs. & Wktrc. & Obs. & Wktrc. & Obs. & Wktrc. & Obs. & Wktrc. \\
 \hline
 1 & 0.211  & 0.197 & 0.284 & 0.283 & 0.328 & 0.232 & - & - & - & - \\
 2 & 0.199 & 0.206 & 0.285 & 0.280 & 0.318 & 0.199 & - & - & - & - \\
 10 & 0.595 & 0.204 & 4.767 & 0.789 & 0.325 & 0.203 & - & - & - & - \\
 11 & - & - & - & - & - & - & 0.519 & 0.223 & 0.571 & 0.221 \\
 14 & - & - & -  & - & - & - & 1.90 & 0.258 & 2.11 & 0.271 \\
 15 & 14.351 & 0.575 & 1130.12 & 15.353 & 0.320 & 0.206  & 3.555 & 0.312 & 4.008 & 0.295 \\
 17 & -  & -  & - & - & - & - & 13.114  & 0.657 & 15.430  & 0.750   \\
 19 & 252.80 & 9.158 & $>$7200 & 271.92 & 0.338 & 0.199  & 52.449  & 2.216 & 64.555 & 1.976 \\
 \hline
\end{tabular}
\end{table}
Our approach is focused on validating detection completeness with respect to explicit, modeled attack scenarios. It is not designed for zero-day attacks or for adversarial behaviors that do not manifest in the threat model, such as chains involving only benign system actions. Integration with anomaly-based detection or dynamic threat modeling could enhance resilience against such evasive techniques, but this remains outside the current scope.

\section{Conclusion}
We presented a formal verification approach that bridges attack modeling and threat detection by integrating attack trees, GTDL, and CADP/LNT. Using bisimulation and weak trace inclusion, our methodology validates the conformance of detection rules to threat models.
While our approach systematically validates detection coverage, 
its guarantees remain bounded by the quality of the underlying attack trees. 
Constructing such models is resource-intensive and requires domain expertise. 
Thus, formal verification cannot compensate for incomplete or inaccurate 
threat modeling, though it can highlight coverage gaps once models are provided.

Case studies on real-world malware (LokiBot, Emotet) show improved precision and automated coverage analysis, while parametric benchmarks confirm the scalability of weak trace inclusion.

Our approach is architecture-agnostic and can be integrated into SIEM, XDR, or IDS platforms. Its main limitation lies in the availability and quality of attack trees, which remain difficult to construct and often require expert effort. Nevertheless, automating the verification pipeline offers a practical way to refine detection rules before deployment, improving accuracy and coverage while acknowledging that assurance ultimately depends on the fidelity of the underlying threat models.

\bibliographystyle{eptcs}
\bibliography{biblio}

\appendix
\section{Technical Listings}
\label{sec:appendix_listings}

This appendix collects the GTDL signatures and their LNT counterparts used
in our LokiBot case-study.

\paragraph{GTDL partial rule.}
Multiple micro-signatures raise flags that a final rule aggregates:

\begin{flushleft}\footnotesize
\begin{lstlisting}
[DETECTION] Detection_name = 'LokibotProcess'  Apply_when = "Process"
[RULE]
v_process  = inPluginCall(IsProcessName, "ytpgwim");
v_location = inPluginCall(IsInProcessPath, "%TEMP%");
IF v_process AND v_location THEN
    GlobalFlag.Set("LokibotProcess"); 
END IF
\end{lstlisting}
\end{flushleft}

\paragraph{GTDL incident rule.}

\begin{flushleft}\footnotesize
\begin{lstlisting}
[DETECTION] Detection_name = 'LokibotIncident' Apply_when = "GlobalFlags"
[RULE]
flag1 = GlobalFlag.IsSet("LokibotProcess");
flag2 = GlobalFlag.IsSet("BotExtensions");
flag3 = GlobalFlag.IsSet("TempRunKey");
flag4 = GlobalFlag.IsSet("KnownCCAccesed");
IF flag1 AND flag2 AND flag3 AND flag4 THEN
    GlobalFlag.Set("Lokibot Incident Detected"); 
END IF
\end{lstlisting}
\end{flushleft}

\paragraph{LNT model.}  
Concrete IoC values replace each \texttt{inPluginCall(\dots)} invocation.

\begin{flushleft}\footnotesize
\begin{verbatim}
process LokibotProcess [lokiBotProcSet:FLAG_CHANNEL]
    (in var processName, processPath:String) is
  if processName == "yptgwim" and processPath == "%TEMP%" then
     lokiBotProcSet(TRUE)
  end if
end process
\end{verbatim}
\end{flushleft}

The full engine composes the partial detectors and the incident logic:

\begin{flushleft}\footnotesize
\begin{verbatim}
process Engine [lokiBotProcSet, lokiBotExtset,
   lokiBotTempRunKey:FLAG_CHANNEL, lokiBotDet:any] is
loop
  par
    LokibotProcess     [lokiBotProcSet] ("yptgwim", "%TEMP%") ||
    LokibotExtension   [lokiBotExtset]  (".exe")              ||
    LokiTempExeRunKey  [lokiBotTempRunKey] ("Run","Run")      ||
    LokiDetection [lokiBotProcSet, lokiBotExtset,
                   lokiBotTempRunKey, lokiBotDet]
  end par
end loop
end process
\end{verbatim}
\end{flushleft}

\paragraph{Attack-tree model.}  
An AND tree reproduces the four independent actions plus a confirmation
leaf, yielding a bisimilar LTS:

\begin{figure}[ht]
  \centering
  \includegraphics[scale=0.5]{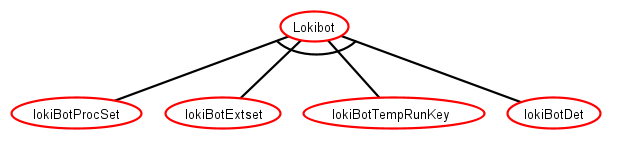}
  \caption{LokiBot attack tree}
  \vspace*{-10pt}
\end{figure}

\begin{flushleft}\footnotesize
\begin{verbatim}
process LokibotTree [lokiBotProcSet, lokiBotExtset,
   lokiBotTempRunKey:FLAG_CHANNEL, lokiBotDet:any] is
par
  LokibotProcessLeaf   [lokiBotProcSet]      ||
  LokibotExtensionLeaf [lokiBotExtset]       ||
  LokiTempExeRunKeyLeaf[lokiBotTempRunKey]   ||
  LokibotActions [lokiBotProcSet, lokiBotExtset,
                  lokiBotTempRunKey, lokiBotDet]
end par
end process
\end{verbatim}
\end{flushleft}

\begin{flushleft}\small
\begin{verbatim}
process LokibotProcessLeaf [lokiBotProcSet:FLAG_CHANNEL] is
  lokiBotProcSet(TRUE)
end process
\end{verbatim}
\end{flushleft}

\paragraph{Verification outcome.}
CADP reports \texttt{TRUE}: the detection and attack-tree LTSs are
observationally equivalent.  Runtime was \(0.234\) s for bisimulation
and \(0.254\) s for weak trace inclusion, confirming negligible overhead
for this single-AND case.

\end{document}